\documentclass[a4paper,12pt]{article}

\makeatletter
 
  \@addtoreset{equation}{section}
 \makeatother

\usepackage{amsfonts}
\usepackage{amssymb}
\usepackage{mathrsfs}
\usepackage{amsmath,amsthm}
\usepackage{amsmath}

\usepackage{setspace}

\usepackage{cite}

\usepackage[dvipdfmx]{graphicx}
\usepackage[dvipdfmx]{color}

\usepackage{here}

\usepackage{indentfirst}

\begin{document}

\newtheorem{theorem}{Theorem}
\newtheorem{lemma}{Lemma}
\newtheorem{proposition}{Proposition}
\newtheorem{cor}{Corollary}
\theoremstyle{definition}
\newtheorem{defn}{Definition}
\newtheorem{remark}{Remark}
\newtheorem{step}{Step}

\newcommand{\Cov}{\mathop {\rm Cov}}
\newcommand{\Var}{\mathop {\rm Var}}
\newcommand{\E}{\mathop {\rm E}}
\newcommand{\const }{\mathop {\rm const }}
\everymath {\displaystyle}

\newcommand{\ruby}[2]{
\leavevmode
\setbox0=\hbox{#1}
\setbox1=\hbox{\tiny #2}
\ifdim\wd0>\wd1 \dimen0=\wd0 \else \dimen0=\wd1 \fi
\hbox{
\kanjiskip=0pt plus 2fil
\xkanjiskip=0pt plus 2fil
\vbox{
\hbox to \dimen0{
\small \hfil#2\hfil}
\nointerlineskip
\hbox to \dimen0{\mathstrut\hfil#1\hfil}}}}

\def\qedsymbol{$\blacksquare$}

\renewcommand{\refname }{References}

\title{
An optimal execution problem with S-shaped market impact functions
\footnote{Forthcoming in {\it Communications on Stochastic Analysis}.}
}

\author{Takashi Kato
\footnote{Association of Mathematical Finance Laboratory (AMFiL), 
              2--10, Kojimachi, Chiyoda, Tokyo 102-0083, Japan, 
E-mail: \texttt{takashi.kato@mathfi-lab.com}}
}

\date{First Version: June 28, 2017\\
This Version: October 2, 2017}

\maketitle 

\begin{abstract}
In this study, we extend the optimal execution problem with convex market impact function studied in Kato \cite {Kato_FS} to the case where 
the market impact function is S-shaped, that is, concave on $[0, \bar {x}_0]$ and convex on $[\bar {x}_0, \infty )$ for some $\bar {x}_0 \geq 0$. 
We study the corresponding Hamilton--Jacobi--Bellman equation and show that 
the optimal execution speed under the S-shaped market impact is equal to zero or larger than $\bar {x}_0$. 
Moreover, we provide some examples of the Black--Scholes model. 
We show that 
the optimal strategy for a risk-neutral trader with small shares is the time-weighted average price strategy 
whenever the market impact function is S-shaped. 
\end{abstract}

\everymath {\displaystyle}

\section{Introduction}
Optimal execution problems have been widely investigated in mathematical finance as a type of stochastic control problem. 
There are various studies of optimal execution, such as \cite {Alfonsi-Fruth-Schied, Almgren-Chriss, Bertsimas-Lo, Gatheral-Schied_AC1, Obizhaeva-Wang} and references therein, 
and Gatheral and Schied \cite {Gatheral-Schied} survey several dynamic models of optimal execution. 
To study this type of problem, we cannot ignore market impact (MI), which is a market liquidity problem. 
Here, we consider a situation where a single trader has many shares of a security and tries to sell (liquidate) it until a time horizon. 
A large selling order induces a gap between supply and demand, causing a decrease in the security price. 
This effect is called the MI, and the trader should reduce the liquidation speed to avoid the MI cost. 
However, reducing the liquidation speed also increases the timing cost, which is caused by the random fluctuation of the security price over time. 
The trader should optimize the execution strategy by considering the MI cost and the timing cost. 
Therefore, the MI function, $g(x)$, plays an important role in studying optimal execution problems. 
Here, $g(x)$ implies the decrease of the security price by selling $x$ shares (or selling rate). 

The simplest setting for $g$ is a linear function. 
For instance, in \cite {Almgren-Chriss, Bertsimas-Lo, Obizhaeva-Wang}, optimal execution problems are treated mainly with linear MI functions and derive optimal execution strategies. 
However, there are studies on optimization problems with non-linear $g$ \cite {Alfonsi-Fruth-Schied, Gueant, Ishitani-Kato_COSA1, Ishitani-Kato_COSA2, Kato_FS, Kato_TJSIAM}. 
In particular, we derive a mathematically adequate continuous-time model of an optimal execution problem as a limit of discrete-time optimization problems in \cite {Kato_FS} when $g$ is strictly convex. 

It is still unclear what form of $g$ is natural. 
Recently, it has been proposed that an S-shaped function is suitable for $g$. 
That is, $g(x)$ should be concave on $[0, \bar {x}_0]$ and convex on $[\bar {x}_0, \infty )$ for some $\bar {x}_0 \geq 0$. 
Many traders intuitively expect that MI functions are S-shaped \cite {Kato_TJSIAM}. 
Moreover, in \cite {Rosu}, we find an empirical prediction for a hump-shaped limit order book, which corresponds to the S-shaped MI function. 
Therefore, our previous study \cite {Kato_FS} should be extended to include the S-shaped MI function, $g$. 
We have tackled this problem partially in \cite {Kato_TJSIAM}, but there are still many mathematical and financial questions at this stage. 
For instance, when we consider the optimal execution problem with S-shaped $g$, 
it is intuitive that the optimal execution speed should not reach the range $(0, \bar {x}_0]$. 
However, we have not proved this finding mathematically. 
Moreover, we have not discussed the Hamilton--Jacobi--Bellman (HJB) equations corresponding to our optimization problem sufficiently. 

In this paper, we resolve these questions as a continuation of our previous study \cite {Kato_TJSIAM}. 
We completely generalize our previous results \cite {Kato_FS} to the case of S-shaped $g$ and study the above questions. 
Moreover, we find that the optimal execution strategy of a risk-neutral trader in the Black--Scholes market model is 
the time-weighted average price (TWAP) strategy, that is, to sell at a constant speed. 
This is the same result as Theorem 5.4(ii) in \cite {Kato_FS} when $g$ is a quadratic function; 
however, it is not necessary to assume an explicit form of $g$. 
We show that this result is true whenever $g$ is S-shaped. 
This result generalizes Theorem 5.4(ii) in \cite {Kato_FS} 
and provides an analytical solution to the optimal execution problem with an uncertain MI given in Section 5.2 of \cite {Ishitani-Kato_COSA2}. 

The rest of this paper is as follows. 
In Section \ref {sec_model}, we introduce a mathematical model of an optimization problem based on our previous work \cite {Kato_TJSIAM} and review the previous results. 
In Section \ref {sec_viscosity}, we characterize our value function as a viscosity solution to the corresponding HJB equation. 
We show the uniqueness of the viscosity solutions to the HJB equation under adequate conditions. 
To investigate the properties of optimal strategies, we introduce a verification theorem and show that the optimal execution strategy does not take the value in $(0, \bar{x}_0]$ in Section \ref {sec_verification}. 
In Section \ref {sec_eg}, we present some examples. 
In particular, we demonstrate the robustness of the TWAP strategy as an optimal strategy in the Black--Scholes model with general shaped MI functions. 
We summarize our argument and introduce future tasks in Section \ref {sec_conclusion}. 
Section \ref {sec_supplement} gives supplemental arguments to guarantee consistency between our present model and our previous model \cite {Kato_TJSIAM}. 
All proofs are in Section \ref {sec_proof}.

\section{Model Settings}\label{sec_model}

Let $(\Omega , \mathcal {F}, (\mathcal {F}_t)_{0\leq t\leq T}, P)$ be a stochastic basis and let $(B_t)_{0\leq t\leq T}$ be a 
one-dimensional Brownian motion ($T > 0$). 
Set $D = \mathbb {R}\times [0, \infty )^2$ and denote by 
$\mathcal {C}$ the set of non-decreasing, non-negative, and continuous functions with polynomial growth defined on $D$. 
For $u\in \mathcal {C}$, 
we define a function, $J(\cdot \ ; u) : [0, T]\times D\longrightarrow \mathbb {R}$, as
\begin{eqnarray}\label{def_J}
J(t, c, x, s ; u) = \sup _{(x_r)_r\in \mathcal {A}_t(x)}\E [u(C_t, X_t, S_t)], 
\end{eqnarray}
where $(C_r)_r, (X_r)_r$, and $(S_r)_r$ are stochastic processes given by 
\begin{eqnarray}\label{notation_SDE}\nonumber 
dC_r &=& x_rS_rdr, \\\nonumber 
dX_r &=& -x_rdr, \\
dS_r &=& \hat{b}(S_r)dr + \hat{\sigma }(S_r)dB_r - S_rg(x_r)dr, 
\end{eqnarray}
and $(C_0, X_0, S_0) = (c, x, s)$, and 
$\mathcal {A}_t(x)$ is the set of non-negative $(\mathcal {F}_r)_r$-progressively measurable process, $(x_r)_{0\leq r\leq t}$, 
satisfying $\int ^t_0x_rdr \leq x$ a.s. 
We call an element of $\mathcal {A}_t(x)$ an admissible strategy. 
Here, $\hat{b}$ and $\hat{\sigma }$ are defined as 
\begin{eqnarray*}
\hat{b}(s) = \left( b(\log s) + \frac{1}{2}\sigma (\log s)^2\right) s, \ \ \hat{\sigma }(s) = \sigma (\log s)s, \ \ s > 0
\end{eqnarray*}
and $\hat{b}(0) = \hat{\sigma }(0) = 0$, where 
$b, \sigma : \mathbb {R}\longrightarrow \mathbb {R}$ are bounded and Lipschitz continuous functions. 
$g\in C([0, \infty ))\cap C^1((0, \infty ))$ is a non-negative function with $g(0) = 0$. 

Function $J$ implies the value function of an optimal execution problem with MI function $g$ and 
is derived as a limit of discrete-time value functions in \cite {Kato_FS} when $g$ is convex, 
and in \cite {Kato_TJSIAM} when $g$ is S-shaped, that is, when $h := g'$ satisfies the following conditions.
\begin{itemize}
 \item [{[A1]}] $h(x)\geq 0$, $x > 0$.
 \item [{[A2]}] $\lim _{x\rightarrow 0}xh(x) = 0$. 
 \item [{[A3]}] There is an $\bar {x}_0\geq 0$ such that $h$ is strictly decreasing on $(0, \bar {x}_0]$ and strictly increasing on $[\bar {x}_0, \infty )$. 
 \item [{[A4]}] $h(\infty ) = \lim _{x\rightarrow 0}h(x) = \infty $. 
\end{itemize}
Condition [A3] implies that $g$ is concave on $[0, \bar {x}_0]$ and convex on $[\bar {x}_0, \infty )$. 
In this paper, we always assume [A1]--[A4]. 

We briefly introduce the financial implications of our model 
(see \cite {Ishitani-Kato_COSA1, Ishitani-Kato_COSA2, Kato_FS} for more details). 
We assume that there is a single trader who has many shares $x_0$ of a security whose price is $s_0$ at the initial time. 
The trader tries to sell the security in the market until time horizon $T$, but the selling behavior affects the security price via the effect of MI (denoted as the term $-g(x_r)dr$ in (\ref {SDE_Y})). 
$S_r$ is the security price at time $r$, and $C_r$ (resp., $X_r$) describes the cash amount (resp., shares of the security) held at time $r$. 
The trader's purpose is to maximize the terminal expected utility, $J(T, c_0, x_0, s_0 ; u) = \E [u(C_T, X_T, S_T)]$, by 
controlling an execution strategy, $(x_r)_r\in \mathcal {A}_T(x)$. 
Here, $x_r$ implies the liquidation speed at time $r$; in other words, the trader sells $x_rdr$ amount in the infinitesimal time interval $[r, r + dr]$. 
To solve this problem, we introduce the value function $J(t, c, x, s ; u)$ for each $t, c, x$, and $s$ to apply the dynamic programming method. 

\begin{remark} \ 
\begin{itemize}
 \item [(i)] The log-price process $Y_r = \log S_r$ satisfies the stochastic differential equation (SDE), 
\begin{eqnarray}\label{SDE_Y}
dY_r = b(Y_r)dr + \sigma (Y_r)dB_r - g(x_r)dr, 
\end{eqnarray}
whenever $S_r > 0$, $r \geq 0$. 
 \item [(ii)] In stochastic control theory,
$(x_r)_r\in \mathcal {A}_t(x)$ is called a control process and $(C_r, X_r, S_r)_r$ defined in (\ref {notation_SDE}) is its controlled process. 
However, the existence and uniqueness of $(C_r, X_r, S_r)_r$ for each $(x_r)_r$ is not obvious in our case, 
because $(Y_r)_r$ may diverge due to the term $-g(x_r)dr$. 
We can overcome this difficulty by regarding $S_r = 0$ after $Y_r$ diverges to $-\infty $ (see Section \ref {sec_supplement} for details). 
 \item [(iii)] In \cite {Ishitani-Kato_COSA1, Kato_FS, Kato_TJSIAM}, 
we require an additional assumption such that each admissible strategy is essentially bounded; that is, we consider the optimization problem 
\begin{eqnarray}\label{def_J_infty}
J^\infty (t, c, x, s ; u) = \sup _{(x_r)_r\in \mathcal {A}^\infty _t(x)}\E [u(C_t, X_t, S_t)] 
\end{eqnarray}
instead of (\ref {def_J}), where 
\begin{eqnarray*}
\mathcal {A}^\infty _t(x) = \left\{ (x_r)_r\in \mathcal {A}_t(x)\ ; \ \mathop {\rm esssup}_{r, \omega }x_r(\omega ) < \infty   \right\} . 
\end{eqnarray*}
This condition arises in the process of taking the limit from the discrete-time model to the continuous-time model; 
however, it is a mathematical technical condition and is unnatural in relation to finance. 
We can show that $J$ coincides with $J^\infty $, and thus 
we are not overly concerned about this problem (also see Section \ref {sec_supplement}). 
\end{itemize}
\end{remark}

In \cite {Kato_TJSIAM}, we show that $J(\cdot \ ; u)$ is continuous on $[0, T]\times D$ and 
$J(r, \cdot \ ; u)\in \mathcal {C}$ for each $r\geq 0$ and $u\in \mathcal {C}$. 
Moreover, $J$ satisfies the dynamic programming principle,
\begin{eqnarray*}
J(t+r, c, x, s ; u) = J(t, c, x, s ; J(r, \cdot ; u)),
\end{eqnarray*}
for each $(c, x, s)\in D$, $u\in \mathcal {C}$ and $t, r\geq 0$ with $t + r\leq T$. 

By using these results, we characterize $J$ as a viscosity solution of the corresponding HJB equation in the next section. 
From now on, we fix $u\in \mathcal {C}$ and denote $J(t, c, x, s ; u) = J(t, c, x, s)$ for brevity. 

\section{Main Results I: Viscosity Properties}\label{sec_viscosity}

Our first main result is as follows. 

\begin{theorem}\label{th_viscosity} \ 
\begin{itemize}
 \item [$({\rm i})$] We assume that 
\begin{eqnarray}\label{ass_viscosity}
\liminf _{\varepsilon \rightarrow 0}\frac{1}{\varepsilon } 
\left( J(t, c, x, s + \varepsilon ) - J(t, c, x, s)\right) > 0, \ \ 
(t, c, x, s)\in (0, T]\times \tilde{D}, 
\end{eqnarray}
where $\tilde {D} = {\rm int}D = \mathbb {R}\times (0, \infty )^2$. 
Then, $J$ is a viscosity solution of the following HJB equation on $(0, T]\times \tilde{D}$:
\begin{eqnarray}\label{HJB}
\frac{\partial }{\partial t}J - \sup _{y\geq 0}\mathscr {L}^yJ = 0, 
\end{eqnarray}
where 
\begin{eqnarray*}
\mathscr {L}^y = (\hat{b}(s) - sg(y))\frac{\partial }{\partial s} + 
\frac{1}{2}\hat{\sigma }(s)^2\frac{\partial ^2}{\partial s^2} + 
y\left( s\frac{\partial }{\partial c} - \frac{\partial }{\partial x} \right) . 
\end{eqnarray*}
 \item [$({\rm ii})$] We assume $(\ref {ass_viscosity})$, that $\hat{b}$ and $\hat{\sigma }$ are Lipschitz continuous, and that $\liminf _{x\rightarrow \infty }\allowbreak h(x)/x > 0$. 
Then, we see the uniqueness of a viscosity solution of $(\ref {HJB})$ in the following sense.
If a continuous function $v: [0, T]\times D\longrightarrow \mathbb {R}$ with polynomial growth 
is a viscosity solution of $(\ref {HJB})$ and satisfies the boundary conditions, 
\begin{eqnarray}\label{boundary_cond1}
v(0, c, x, s) &=& u(c, x, s), \\\label{boundary_cond2}
v(t, c, 0, s) &=& \E [u(c, 0, Z_t(s))], \\\label{boundary_cond3}
v(t, c, x, 0) &=& u(c, x, 0), 
\end{eqnarray}
then it holds that $J = v$, where 
$(Z_r(s))_r$ is a unique solution to the SDE: 
\begin{eqnarray}\label{Z}
dZ_r(s) = \hat{b}(Z_r(s))dr + \hat{\sigma }(Z_r(s))dB_r, \ \ Z_0(s) = s. 
\end{eqnarray}
\end{itemize}
\end{theorem}

\begin{remark} \ 
\begin{itemize}
 \item [(i)] The assertions of Theorem \ref {th_viscosity} are the same as those of Theorems 3.3 and 3.6 in \cite {Kato_FS}. 
Thus, Theorem \ref {th_viscosity} was already obtained when $g$ is convex, namely, when $\bar {x}_0 = 0$. 
As mentioned in Remark 3.7 of \cite {Kato_FS}, 
our HJB equation (\ref {HJB}) does not satisfy standard assumptions to apply a standard argument to viscosity characterization discussed in, for instance, 
\cite {DaLio-Ley, Fleming-Soner, Krylov, Nagai}. 
We demonstrate Theorem \ref {th_viscosity}(i) through a refinement of the proof of Theorem 3.3 in \cite {Kato_FS}. 
In the proof of Theorem 3.6 in \cite {Kato_FS}, 
we do not use the convexity of $g$ mainly, so Theorem \ref {th_viscosity}(ii) is obtained in a similar way to the proof of Propositions B.21--B.23 in \cite {Kato_FS}. 
 \item [(ii)] The following condition is a standard natural condition for a utility function in mathematical finance: 
\begin{itemize}
 \item [{[B]}] $u(c, x, s) = U(c)$ for some concave function $U\in C^1(\mathbb {R})$. 
\end{itemize}
Under [B], the boundary conditions (\ref {boundary_cond1})--(\ref {boundary_cond3}) are simplified as
\begin{eqnarray*}
v(0, c, x, s) = v(t, c, 0, s) = v(t, c, x, 0) = U(c). 
\end{eqnarray*}
 \item [(iii)] It is not easy to check (\ref {ass_viscosity}) in general. 
When $g$ is convex, the natural and simple sufficient conditions of (\ref {ass_viscosity}) are introduced in \cite {Kato_FS} as 
\begin{itemize}
 \item [{[C1]}] $u$ satisfies [B]. 
Moreover, it holds that $U'(c)\geq \delta $, $c\in \mathbb {R}$, for some $\delta > 0$. 
 \item [{[C2]}] $b$ and $\sigma $ are differentiable and their derivatives are Lipschitz continuous and uniformly bounded. 
\end{itemize}
In our case, by the same proof as for Proposition 3.5 in \cite {Kato_FS}, 
we also verify that (\ref {ass_viscosity}) holds under [C1]--[C2]. 
\end{itemize}
\end{remark}

\section{Main Results II: Verification Arguments}\label{sec_verification}

Theorem 7.4 in \cite {Kato_TJSIAM} gives us a typical example where 
an optimal execution strategy takes the value zero or larger than $\bar {x}_0$. 
This result is consistent with financial intuition, such as selling with the speed in the range of the concave part of $g$ (i.e., $(0, \bar {x}_0]$) 
induces superfluous transaction cost. 
In this section, we present a verification theorem to demonstrate that the
optimal execution speed is in $\{0\}\cup (\bar {x}_0, \infty )$ in general. 

First, we introduce notation to state our second main result. 
Conditions [A3] and [A4] show that there is an inverse function, $h^{-1} : [h(\bar {x}_0), \infty ) \longrightarrow [\bar {x}_0, \infty )$, of $h$. 
Then we define 
\begin{eqnarray}\label{def_H}
\mathcal {H}(s, p) &=& 
\frac{sp_c - p_x}{sp_s}1_{\{sp_s > 0\}}, \\\label{def_Xi}
\Xi (s, p) &=& h^{-1}(\mathcal {H}(s, p))1_{\Lambda }(s, p) 
\end{eqnarray}
for $s\geq 0$ and $p = (p_c, p_x, p_s)'\in \mathbb {R}^3$, where 
\begin{eqnarray*}
\Lambda = \{ (s, p)\in (0, \infty )\times \mathbb {R}^3 &;& 
p_s> 0, \mathcal {H}(s, p) > h(\bar {x}_0), \\
&& g(h^{-1}(\mathcal {H}(s, p))) < \mathcal {H}(s, p)h^{-1}(\mathcal {H}(s, p)) \},
\end{eqnarray*}
and $A'$ denotes the transpose of $A$. 
Moreover, for each continuously differentiable function, $v : [0, T] \times D\longrightarrow \mathbb {R}$, we define 
\begin{eqnarray*}
\bar{b}^v(t, c, x, s) = 
\left(
\begin{array}{c}
 s\Xi (s, \mathscr {D}v( T - t, c, x, s ))\\
 -\Xi (s, \mathscr {D}v( T - t, c, x, s ))\\
 \hat{b}(s) - sg(\Xi (s, \mathscr {D}v( T - t, c, x, s )))
\end{array}
\right), 
\end{eqnarray*}
where $\mathscr{D} = \left (\frac{\partial }{\partial c}, \frac{\partial }{\partial x}, \frac{\partial }{\partial s}\right )'$. 

Now, we present our second main result. 

\begin{theorem}\label{th_S-shaped}We assume {\rm (\ref {ass_viscosity})} and {\rm [B]}, that $J\in C^{1,1,1,2}((0, T]\times D)$, and that for given $(c_0, x_0, s_0)\in \tilde{D}$, there is a continuous process, $(\overline {C}_t, \overline {X}_t, \overline {S}_t)_t$, which satisfies 
\begin{eqnarray}\label{SDE_S-shaped}
d\left(
\begin{array}{c}
 \overline {C}_t\\
 \overline {X}_t\\
 \overline {S}_t
\end{array}
\right) = 
\bar{b}^J(t, \overline {C}_t, \overline {X}_t, \overline {S}_t)dt + 
\left(
\begin{array}{c}
 0\\
 0\\
 \hat{\sigma }(\overline {S}_t)
\end{array}
\right)dB_t, \ \ t\in [0, \bar{\tau }] 
\end{eqnarray}
and $(\overline {C}_0, \overline {X}_0, \overline{S}_0) = (c_0, x_0, s_0)$, 
where 
\begin{eqnarray*}
\bar{\tau } = \inf \{ t\geq 0\ ; \ (\overline {C}_t, \overline {X}_t, \overline {S}_t)\in \partial D \} \wedge T. 
\end{eqnarray*}
Then, there is an optimizer, $(\hat{x}_t)_t$ to $J(T, c_0, x_0, s_0)$, such that 
$\hat{x}_t\in \{ 0 \}\cup (\bar {x}_0, \infty )$, $t\in [0, T]$ a.s. 
\end{theorem}

When executing a large amount of the security, 
it is important to decrease the execution speed to reduce the execution cost. 
However, Theorem \ref {th_S-shaped} tells us that when the MI function is S-shaped (especially, concave on $[0, \bar {x}_0]$), 
it is undesirable to decrease the execution speed beyond the threshold, $\bar {x}_0$. 
An optimal execution strategy in this case is to sell with the execution speed greater than $\bar {x}_0$ or to stop selling.

We can apply Theorem \ref {th_S-shaped} if we verify the smoothness of the value function $J$. 
Even if we find a classical (sub)solution of (\ref {HJB}), which does not necessarily satisfy the boundary conditions (\ref {boundary_cond1})--(\ref {boundary_cond3}), 
we can construct an optimal strategy to $J(T, c_0, x_0, s_0)$. 
We introduce the following verification theorem.

\begin{theorem}\label{th_verification}
Let $(c_0, x_0, s_0)\in \tilde{D}$ and let $v\in C([0, T]\times D)\cap C^{1,1,1,2}((0, T]\times \tilde{D})$ be a function that satisfies the following conditions.
\begin{itemize}
 \item [{\rm (i)}] There are $K, m > 0$ such that 
\begin{eqnarray*}
|v(t, c, x, s)| \leq K(1 + c^m + x^m + s^m), \ \ 
t\in [0, T], \ (c, x, s)\in D. 
\end{eqnarray*}
 \item [{\rm (ii)}] $v(0, c, x, s)\geq u(c, x, s)$ holds for each $(c, x, s)\in D$. 
 \item [{\rm (iii)}] $\frac{\partial }{\partial t}v - \sup _{y\geq 0}\mathscr {L}^yv \geq 0$ on $(0, T]\times \tilde{D}$. 
 \item [{\rm (iv)}] There is an $(\hat{x}_t)_t\in \mathcal {A}_T(x_0)$ such that 
$\E [u(\hat{C}_T, \hat{X}_T, \hat{S}_T)] \geq v(T, c_0, x_0, s_0)$, where 
$(\hat{C}_t, \hat{X}_t, \hat{S}_t)_t$ is given by $(\ref {notation_SDE})$ with $(\hat{C}_0, \hat{X}_0, \hat{S}_0) = (c_0, x_0, s_0)$. 
\end{itemize}
Then, we have $J(T, c_0, x_0, s_0) = \E [u(\hat{C}_T, \hat{X}_T, \hat{S}_T)] = v(T, c_0, x_0, s_0)$ and $(\hat{x}_t)_t$ is its optimizer. 
\end{theorem}

In the next section, we introduce some examples in which we derive optimal execution strategies by using Theorem \ref {th_verification}. 

\section{Examples}\label{sec_eg}

Similar to Section 5 in \cite {Kato_FS}, 
we introduce some examples where the security price process is given as the Black--Scholes model. 
We assume that $b(\cdot )\equiv \mu $ and $\sigma (\cdot ) \equiv \sigma $ are constants 
and the utility function is set as $u_{\mathrm {RN}}(c, x, s) = c$; that is, the trader is risk-neutral. 

By using the same argument as the proof of Proposition 5.2 in \cite {Kato_FS}, we have the following theorem.

\begin{theorem}\label{th_eg}We have $J(t, c, x, s) = c + sW(t, x)$, where 
\begin{eqnarray*}
W(t, x) &=& \sup _{(x_r)_r\in \mathcal {A}^{\rm stat}_t(x)}\int ^t_0\exp \left( -\tilde{\mu }r - \int ^r_0g(x_v)dv \right) x_rdr, \\
\mathcal {A}^{\rm stat}_t(x) &=& \{ (x_r)_r\in \mathcal {A}_t(x)\ ; \ (x_r)_r \mbox { is deterministic} \}, \\
\tilde{\mu } &=& -\mu - \frac{1}{2}\sigma ^2. 
\end{eqnarray*}
\end{theorem}

Theorems \ref {th_viscosity} and \ref {th_eg} lead us to 
\begin{theorem}\label{th_eg_HJB}
$W$ is a viscosity solution to the partial differential equation
\begin{eqnarray}\label{HJB_W}
\frac{\partial }{\partial t}W + \tilde{\mu }W + \inf _{y\geq 0}\left\{ Wg(y) - \left( 1 - \frac{\partial }{\partial x}W\right) y \right\} = 0
\end{eqnarray}
with the boundary condition 
\begin{eqnarray}\label{boundary_W}
W(t, 0) = W(0, x) = 0. 
\end{eqnarray}
Moreover, if $\liminf _{x\rightarrow \infty }h(x)/x > 0$, a viscosity solution to $(\ref {HJB_W})$--$(\ref {boundary_W})$ is unique in the following sense. 
If a continuous function, $w$, with polynomial growth is a viscosity solution to $(\ref {HJB_W})$--$(\ref {boundary_W})$, then $W = w$. 
\end{theorem}

Until the end of this section, we assume $\tilde{\mu } > 0$ to focus on the case where the expected security price decreases over time.

\subsection{Mixed Power MI Function}

Here, we consider the case where $g$ is given by 
\begin{eqnarray}\label{mixed_power_MI}
g(x) = \beta x^{\tilde{\pi }} \ (0\leq x\leq \bar {x}_0), \ \ \alpha x^\pi + \gamma  \ (x > \bar {x}_0)
\end{eqnarray}
for some $\bar {x}_0 \geq 0$, $\alpha > 0$ and $0 < \tilde{\pi } < 1 < \pi $. 
Because $g$ is continuously differentiable, $\beta $ and $\gamma $ must satisfy 
\begin{eqnarray*}
\beta = \frac{\pi }{\tilde{\pi }}\alpha \bar {x}_0^{\pi - \tilde{\pi }}, \ \ 
\gamma = \left( \frac{\pi }{\tilde{\pi }} - 1\right) \alpha \bar {x}_0^\pi . 
\end{eqnarray*}
Figure \ref {fig_mixed_power} shows the form of $g$ when $\tilde{\pi } = 0.5$ and $\pi = 2$. 

\begin{figure}[!h]
\centerline{\includegraphics[height = 8cm,width=12cm]{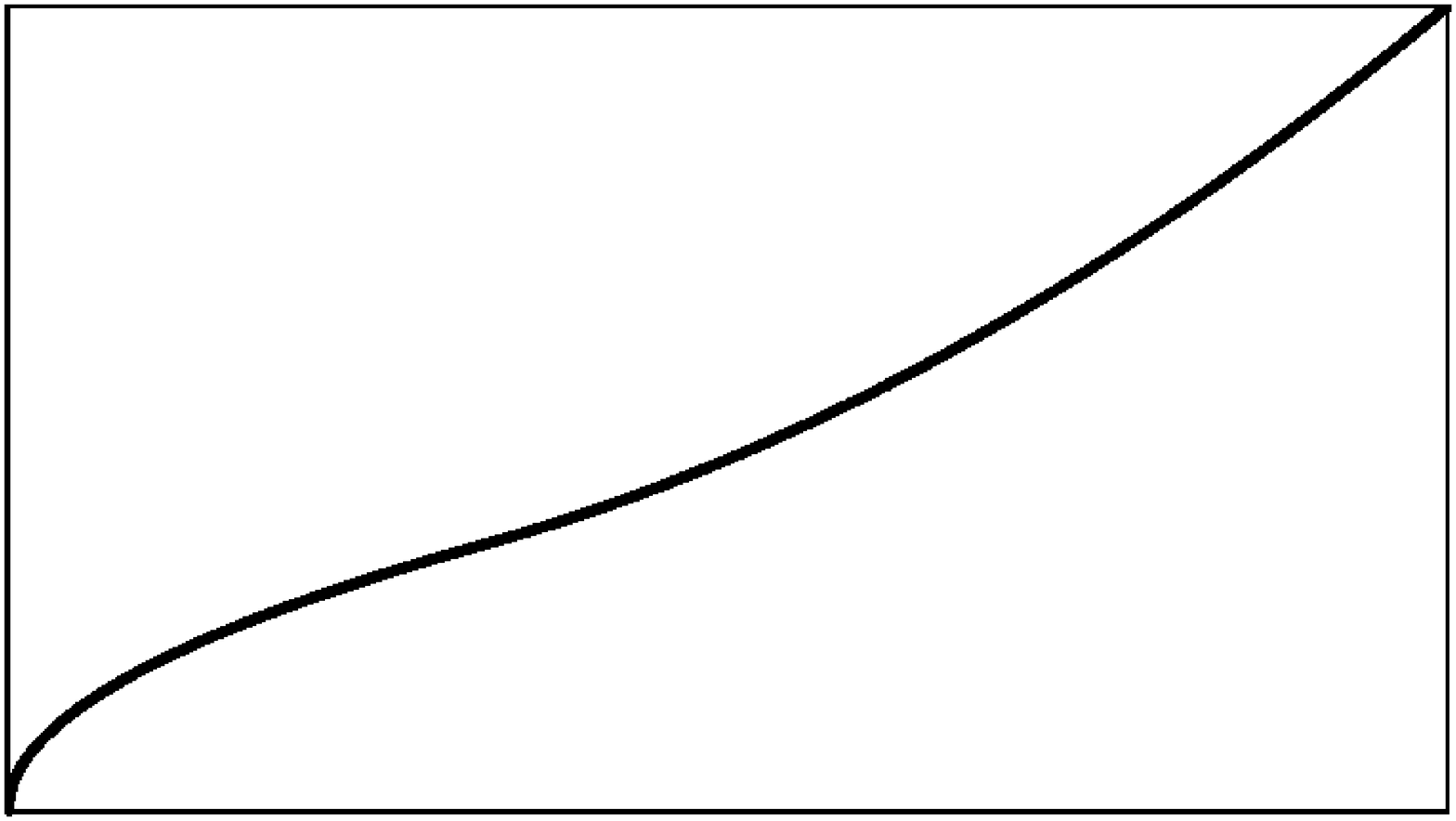}}
\caption{Form of the MI function $g(x)$ defined as (\ref {mixed_power_MI}) with $\tilde{\pi } = 0.5$ and $\pi = 2$. 
The horizontal axis corresponds to $x$. The vertical axis corresponds to $g(x)$. }
\label{fig_mixed_power}
\end{figure}

The next result is a pure extension of Theorem 5.4 in \cite {Kato_FS}. 

\begin{theorem}\label{th_mixed_power}
Set 
\begin{eqnarray*}
x^{*, 1} &=& \frac{1}{\delta _\pi }\mathbb {B}\left ( 1 - \exp \left( - \frac{\pi }{\pi -1}(\tilde{\mu } + \gamma )T \right) ; \frac{1}{\pi } + 1, 2\right ), \\
x^{*, 2} &=& \nu _\pi T, 
\end{eqnarray*}
where 
\begin{eqnarray*}
\delta _\pi = \alpha ^{1/\pi }\pi \left( \frac{\tilde{\mu } + \gamma }{\pi - 1} \right) ^{\frac{\pi - 1}{\pi }}, \ \ 
\nu _\pi = \left( \frac{\tilde{\mu } + \gamma }{(\pi - 1)\alpha } \right) ^{1/\pi }
\end{eqnarray*}
and 
\begin{eqnarray*}
\mathbb {B}(z ; a, b) = \int ^z_0\frac{dx}{x^{a-1}(1-x)^{b-1}}
\end{eqnarray*}
is the incomplete Beta function. 
\begin{itemize}
 \item [{\rm (i)}] If $x_0 \geq x^{*, 1}$, we have 
\begin{eqnarray*}
J(T, c_0, x_0, s_0) = c_0 + 
\frac{s_0}{\delta _\pi }\left( 1 - \exp \left( -\frac{\pi }{\pi - 1}(\tilde{\mu } + \gamma )T \right)  \right) ^{\frac{\pi - 1}{\pi }},
\end{eqnarray*}
and its optimizer is given by 
\begin{eqnarray*}
\hat{x}_t = \nu _\pi 
\left( 1 - \exp \left( -\frac{\pi }{\pi - 1}(\tilde{\mu } + \gamma )(T - t) \right)  \right) ^{-1/\pi }. 
\end{eqnarray*}
 \item [{\rm (ii)}] If $x_0 \leq x^{*, 2}$, we have 
\begin{eqnarray*}
J(T, c_0, x_0, s_0) = c_0 + s_0\cdot \frac{1 - e^{-\delta _\pi x_0}}{\delta _\pi }, 
\end{eqnarray*}
and its optimizer is given by 
\begin{eqnarray*}
\hat{x}_t = \nu _\pi 1_{[0, x_0/\nu _\pi ]}(t). 
\end{eqnarray*}
\end{itemize}
\end{theorem}

Similarly to Theorem 5.4 in \cite {Kato_FS}, 
the form of the optimal strategy changes drastically according to the initial shares $x_0$, 
and we do not have an analytical solution when $x^{*, 2} < x_0 < x^{*, 1}$. 
Moreover, when $x_0 \leq x^{*, 2}$, the optimal strategy is the TWAP strategy, that is to sell with constant speed $\nu _\pi $. 
The TWAP strategy is the optimal strategy for the Almgren--Chriss model, 
which is a standard model of optimal execution, for the risk-neutral trader 
\cite {Almgren-Chriss, Gatheral-Schied_AC1, Kato_JSIAM_VWAP, Kato_AC_Preprint}. 
Theorem \ref {th_mixed_power}(ii) is also obtained as a corollary of the result of the next subsection. 
In addition,
\begin{eqnarray*}
\nu _\pi > \left( \frac{\gamma }{(\pi - 1)\alpha } \right) ^{1/\pi } = 
\left( \frac{\pi - \tilde {\pi }}{\tilde {\pi }(\pi - 1)} \right) ^{1/\pi }\bar{x}_0 > \bar{x}_0; 
\end{eqnarray*}
hence we can verify that $\hat{x}_t\in \{ 0 \} \cup (\bar{x}_0, \infty )$ in both cases of Theorem \ref {th_mixed_power}(i)(ii). 
This is consistent with Theorem \ref {th_S-shaped}.

\subsection{TWAP Strategies for Small Amount Execution}

Next, we consider the case where the amount $x_0$ of initial shares is small. 
Here, we do not restrict the form of $g$ without [A1]--[A4]. 

Before stating the result, we prepare the following proposition. 

\begin{proposition}\label{prop_nu}
Set $G_h(x) = xh(x) - g(x)$. 
Then, there is a unique $\nu _h \in (\bar{x}_0, \infty )$ such that $G_h(\nu _h) = \tilde{\mu }$. 
\end{proposition}

\begin{theorem}\label{th_TWAP}If $x_0 \leq \nu _hT$, 
we have 
\begin{eqnarray}\label{J_TWAP}
J(T, c_0, x_0, s_0) = c_0 + s_0\cdot \frac{1 - e^{-h(\nu _h)x_0}}{h(\nu _h)}, 
\end{eqnarray}
and its optimizer is given by 
\begin{eqnarray}\label{TWAP}
\hat{x}_t = \nu _h1_{[0, x_0/\nu _h]}(t). 
\end{eqnarray}
\end{theorem}

This theorem implies the robustness of the optimality of the TWAP strategy for general shaped MI functions. 
When $x_0$ is small, the optimal execution strategy is to sell the security at speed $\nu _h (> \bar{x}_0)$ 
until the time when the remaining shares become zero. 

\begin{remark} \ 
\begin{itemize}
 \item [(i)] As mentioned in Theorems 4.2 and 5.1 in \cite {Kato_FS}, 
when $g(x) = \alpha x$ $(\alpha > 0)$ is given as a linear function, we have 
\begin{eqnarray}\label{J_linear}
J(T, c_0, x_0, s_0) = c_0 + s_0\cdot \frac{1 - e^{-\alpha x_0}}{\alpha },
\end{eqnarray}
and the corresponding nearly optimal execution strategy is a quasi-block liquidation with the initial time; that is, 
$\hat{x}^\delta _t = (x_0/\delta )1_{[0, \delta ]}(t)$ with $\delta \rightarrow 0$. 
This strategy formally corresponds to (\ref {TWAP}) taking the limit $\nu _h\rightarrow \infty $. 
Note that $h(x) \equiv \alpha $; hence (\ref {J_TWAP}) coincides with (\ref {J_linear}). 
 \item [(ii)] Let us consider an extreme case where 
\begin{eqnarray}\label{extreme_MI}
g(x) = \hat{g}(x - \bar{x}_0)1_{[\bar{x}_0, \infty )}(x)
\end{eqnarray}
for some increasing convex function 
$\hat{g}\in C^1([0, \infty ) ; [0, \infty ))$ with $\hat{g}(0) = \hat{g}'(0) = 0$ and $\hat{g}'(\infty ) = \infty $. 
The form of $g(x)$ is shown in Figure \ref {fig_extreme_g} for $\hat{g}(x)$ set as $x^3$. 
In this case, we can completely avoid the MI cost by selling at a speed lower than or equal to $\bar{x}_0$. 
Therefore, the optimal execution strategy seems to be $\tilde{x}_t = \bar{x}_01_{[0, x_0/\bar{x}_0]}(t)$ at a glance. 
Following the strategy, $(\tilde{x}_t)_t$, we get the expected proceeds
\begin{eqnarray*}
\tilde{\mathcal {C}} := s_0\int ^{x_0/\bar{x}_0}_0e^{-\tilde{\mu }t}\bar{x}_0dt = s_0\iota (\tilde{\mu }/\bar{x}_0 ; x_0), 
\end{eqnarray*}
where $\iota (y ; x) = (1 - e^{-xy})/y$. 
However, Theorem \ref {th_TWAP} implies that this strategy is not optimal; the optimal execution speed, $\nu _h$, is strictly greater than $\bar{x}_0$. 
We compare the expected proceeds 
\begin{eqnarray*}
\hat{\mathcal {C}} := J(T, 0, x_0, s_0) = s_0\iota (h(\nu _h) ; x_0)
\end{eqnarray*}
obtained by the optimal strategy, $(\hat{x}_t)_t$, with $\tilde{\mathcal {C}}$ obtained by $(\tilde{x}_t)_t$. 
Let us denote $\hat{G}(x) = x\hat{g}'(x) - \hat{g}(x)$. 
Then, we see that $\hat{G}(\nu _h - \bar{x}_0) > \hat{G}(0) = 0$, and thus 
\begin{eqnarray*}
\tilde{\mu } = G_h(\nu _h) = \hat{G}(\nu _h - \bar{x}_0) + \bar{x}_0h(\nu _h) > \bar{x}_0h(\nu _h). 
\end{eqnarray*}
This implies that $h(\nu _h) < \tilde{\mu } / \bar{x}_0$. 
Because $\iota (\cdot \hspace{1mm}; x_0)$ is decreasing, we have $\hat{\mathcal {C}} > \tilde{\mathcal {C}}$. 
This is because selling at a lower speed increases the execution time and the timing cost. 
The trader should sell with the optimal speed, $\nu _h$, and accept the MI cost. 
\end{itemize}
\end{remark}

\begin{figure}[!h]
\centerline{\includegraphics[height = 8cm,width=12cm]{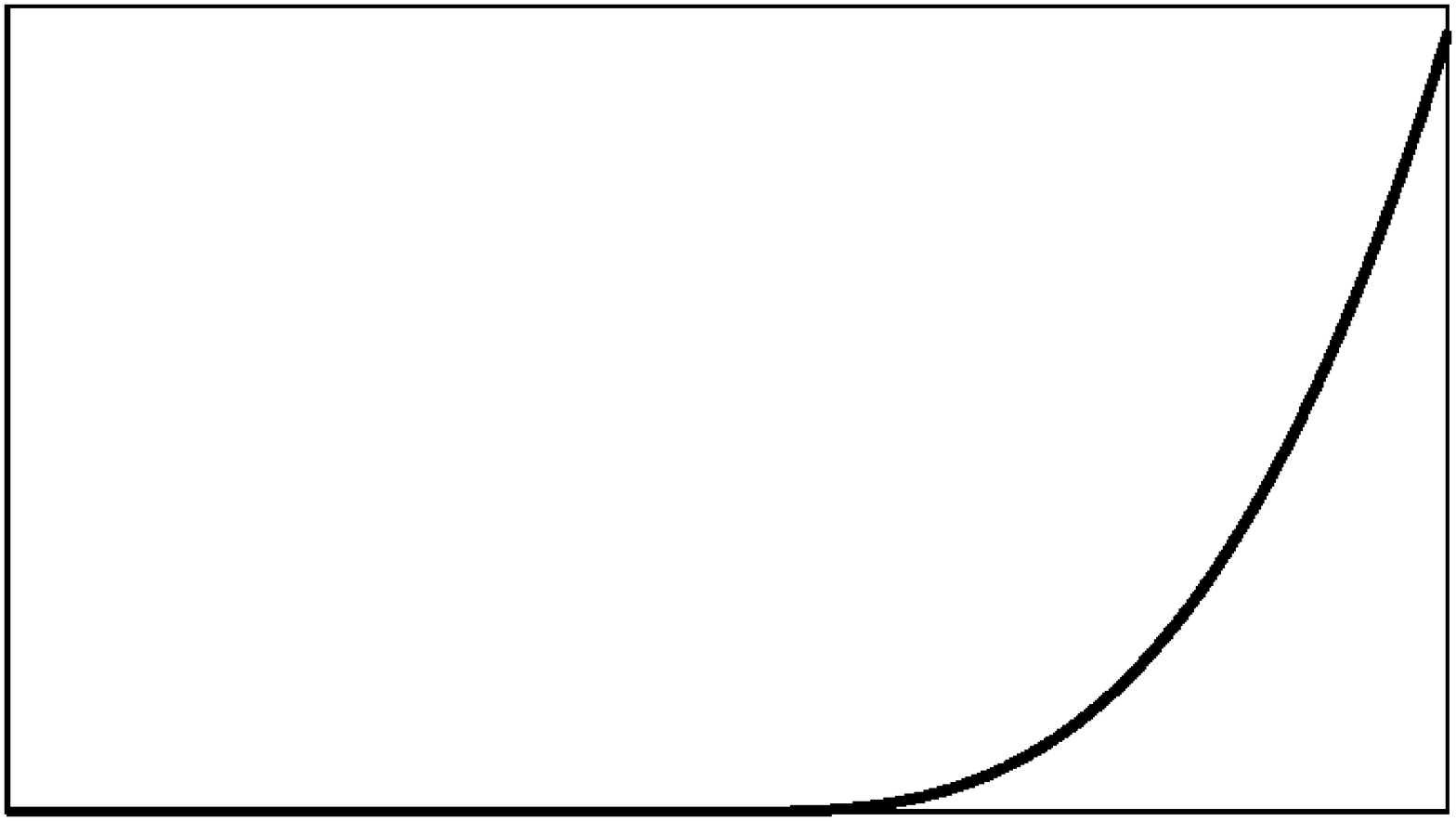}}
\caption{Form of the MI function, $g(x)$, defined as (\ref {extreme_MI}) with $\hat{g}(x) = x^3$. 
The horizontal axis corresponds to $x$. The vertical axis corresponds to $g(x)$. }
\label{fig_extreme_g}
\end{figure}

\subsection{Generalization of a Previous Result in Ishitani and Kato \cite {Ishitani-Kato_COSA2}}

As an application of Theorem \ref {th_TWAP}, 
we provide an analytical solution to an optimal execution problem with uncertain MI studied in 
Section 5.2 of \cite {Ishitani-Kato_COSA2}. 
We consider the optimization problem 
\begin{eqnarray}\label{def_J_Levy}
\sup _{(x_t)_t\in \mathcal {A}_T(x_0)}\E \left [\int ^T_0S_tx_tdt\right ], 
\end{eqnarray}
where $(S_t)_t$ is given by the SDE:
\begin{eqnarray*}
dS_t = S_t(-\tilde{\mu }dt + \sigma dB_t - g(x_t)dL_t), \ \ S_0 = s_0. 
\end{eqnarray*}
Here, $(L_t)_t$ is the L\'evy process, which is independent of $(B_t)_t$ and whose distribution is given by the Gamma distribution
\begin{eqnarray*}
P(L_t - \gamma t \in dz) = \frac{1}{\Gamma (\alpha _1t)\beta _1^{\alpha _1t}}z^{\alpha _1t - 1}e^{-z/\beta _1}1_{(0, \infty )}(z)dz, 
\end{eqnarray*}
where $\alpha _1, \beta _1, \gamma > 0$ satisfy $\alpha _1\beta _1 \leq 8\gamma $ and $\Gamma (z) = \int ^\infty _0t^{z-1}e^{-t}dt$ is the Gamma function. 
Moreover, we assume that $g(x) = \alpha _0x^2$ is given as a quadratic function with $\alpha _0 \geq 0$) 

In Section 5.2 of \cite {Ishitani-Kato_COSA2}, 
we do not find the explicit form of the optimal strategy to (\ref {def_J_Levy}), even when $x_0$ is small. 
However, numerical experiments suggest that the optimal strategy with small $x_0$ is the TWAP strategy. 
Here, we prove mathematically that this conjecture is true. 

\begin{theorem}\label{th_TWAP_Levy}Let $\hat{\nu }$ be the solution to
\begin{eqnarray*}
\gamma \alpha _0\hat{\nu }^2 + 
\alpha _1\left\{ 2\left( 1 - \frac{1}{1 + \alpha _0\beta _1\hat{\nu }^2}\right) - \log (\alpha _0\beta _1\hat{\nu }^2 + 1) \right\}  = \tilde{\mu }. 
\end{eqnarray*}
If $x_0 \leq \hat{\nu }T$, the optimal strategy for $(\ref {def_J_Levy})$ is given by the TWAP strategy
\begin{eqnarray}\label{TWAP_Levy}
\hat{x}_t = \hat{\nu }1_{[0, x_0/\hat{\nu }]}(t). 
\end{eqnarray}
\end{theorem}

\section{Concluding Remarks}\label{sec_conclusion}

In this paper, we studied the optimal execution problem with S-shaped MI functions as a continuation of \cite {Kato_TJSIAM}. 
We showed that our value function is characterized as a viscosity solution of the corresponding HJB equation. 
This is an extended result of that in \cite {Kato_FS}. 
Moreover, we provided the verification theorem to show that 
the optimal execution speed is not in the range $(0, \bar{x}_0]$. 
This implies that the trader should not blindly decrease the execution speed to reduce the MI cost. 

In the Black--Scholes market model, we found that 
an optimal execution strategy is the TWAP strategy when the number of shares of the security held is small. 
A concrete form is not required for the MI function, $g$, so this result is robust and suggests the optimality of TWAP strategy in practice. 

The volume-weighted average price (VWAP) strategy is widely used in trading practice rather than the TWAP strategy \cite {Madhavan}. 
Gatheral and Schied \cite {Gatheral-Schied} pointed out that we should regard the time parameter, $t$, not as physical time but as volume time. 
Volume time implies a stochastic clock, which is measured by a market trading volume process \cite {Ane-Geman, Geman, Mazur, Veraat-Winkel}. 
If we consider the model on a volume time line, we may find the optimality of the VWAP strategy in a similar way to Theorem \ref {th_TWAP}. 
However, we should not ignore the randomness of the market trading volume. 
One of our future tasks is to construct a model of optimal execution with S-shaped MI functions on a volume time line. 

Furthermore, to apply Theorem \ref {th_S-shaped}, we require the value function, $J$, to be smooth, whereas it is difficult to show smoothness in general. 
Moreover, the solvability of SDE (\ref {SDE_S-shaped}) is not clear. 
Further study is needed. 

\appendix 

\section{Supplemental Arguments}\label{sec_supplement}

We present the following propositions, which link the results in our previous study \cite {Kato_TJSIAM} with the present model. 

\begin{proposition}\label{prop_unique_existence}Let $t > 0$ and let $(c, x, s)\in D$. 
For each $(x_r)_{r\leq t}\in \mathcal {A}_t(x)$, 
there is a unique process, $(C_r, X_r, S_r)_{r\leq t}$, that satisfies $(\ref {notation_SDE})$ and $(C_0, X_0, S_0) = (c, x, s)$. 
\end{proposition}

The comparison theorem for solutions of SDEs (see Proposition 5.2.18 in \cite {Karatzas-Shreve} for instance) tells us that 
\begin{eqnarray}\label{est_S_Z}
0\leq S_r \leq Z_r(s) \ \ \mbox{a.s.}, 
\end{eqnarray}
where $(Z_r(s))_r$ is defined in (\ref {Z}). 
Moreover, Lemma B.1 in \cite {Kato_FS} tells us that 
\begin{eqnarray}\label{est_Z_hat}
\E [\sup _{0\leq r\leq t}Z_r(s)^m] < \infty 
\end{eqnarray}
for each $t, s$ and $m > 0$. 
Based on (\ref {est_S_Z})--(\ref {est_Z_hat}), we see that our value function, $J(t, c, x, s)$, is well-defined and finite.

\begin{proposition}\label{prop_compare_value_functions}
$J(t, c, x, s) = J^\infty (t, c, x, s)$. 
\end{proposition}

In \cite {Kato_TJSIAM}, we show some properties of $J^\infty (t, c, x, s)$. 
Proposition \ref {prop_compare_value_functions} implies that these results also hold for $J(t, c, x, s)$.

\section{Proofs}\label{sec_proof}

\begin{proof}[Proof of Proposition \ref {prop_nu}]
First, [A3] implies that 
\begin{eqnarray}\label{est_G1}
G_h(\bar{x}_0) = \bar{x}_0h(\bar{x}_0) - \int ^{\bar{x}_0}_0h(x)dx \leq \bar{x}_0h(\bar{x}_0) - \bar{x}_0h(\bar{x}_0) = 0. 
\end{eqnarray}
[A3] also tells us that 
\begin{eqnarray}\label{est_G2}
G_h(x) - G_h(y) \geq (h(x) - h(y))y > 0
\end{eqnarray}
for each $x > y > \bar{x}_0$. 
Hence, $G_h$ is strictly increasing on $(\bar{x}_0, \infty )$. 
Moreover, letting $x\rightarrow \infty $ in (\ref {est_G2}), we see that 
\begin{eqnarray}\label{est_G3}
\lim _{x\rightarrow \infty }G_h(x) = \infty, 
\end{eqnarray}
owing to condition [A4]. 
Because $G_h$ is continuous on $[\bar{x}_0, \infty )$ and $\tilde{\mu }$ is positive, 
(\ref {est_G1})--(\ref {est_G3}) immediately give the assertion. 
\end{proof}

To show Proposition \ref {prop_unique_existence}, we prepare a lemma. 

\begin{lemma}\label{lem_moment}
Let $(\varphi _t)_t$ be an $(\mathcal {F}_t)_t$-progressively measurable process such that 
$|\varphi _t|\leq K$ for some positive constant $K$. 
Then, there is a $C_{K, T} > 0$ that depends only on $K$ and $T$, such that 
\begin{eqnarray*}
\E \left[ \sup _{0\leq t\leq T}\exp \left( \int ^t_0\varphi _rdB_r \right) \right] \leq C_{K, T}. 
\end{eqnarray*}
\end{lemma}

\begin{proof}
Put 
\begin{eqnarray*}
N_t = \exp \left( \int ^t_0\varphi _rdB_r - \frac{1}{2}\int ^t_0\varphi ^2_rdr \right) - 1. 
\end{eqnarray*}
Ito's formula immediately implies that $(N_t)_t$ is a continuous local martingale starting at $0$ 
and $d\langle N\rangle _t = (N_t + 1)^2\varphi ^2_tdt$. 

Take any $R > 0$ and define 
$\tau _R = \inf \{ t\geq 0\ ; \ \langle N\rangle _t\geq R \}\wedge T$ and 
$m^R_t = \E [\langle N\rangle _{t\wedge \tau _R}] (\leq R < \infty )$. 
Then, we observe 
\begin{eqnarray*}
0\leq m^R_t \leq 2K^2\E \left[ \int ^{t\wedge \tau _R}_0(N^2_r + 1)dr \right]  \leq 2K^2T + 2K^2\int ^t_0m^R_rdr. 
\end{eqnarray*}
We apply the Gronwall inequality to obtain 
\begin{eqnarray*}
m^R_t \leq 2K^2T + 4K^2T^2e^{2K^2T} =: C'_{K, T}. 
\end{eqnarray*}
The Chebyshev inequality implies that $\tau _R \nearrow T$, $R\rightarrow \infty $ a.s., and hence 
$\E [\langle N\rangle _T] = \lim _{R\rightarrow \infty }m^R_t\leq C'_{K, T}$ by the monotone convergence theorem. 
Now we arrive at 
\begin{eqnarray*}
\E \left[ \sup _{0\leq t\leq T}\exp \left( \int ^t_0\varphi _rdB_r \right) \right] \leq 
e^{K^2T/2}(2\E [\langle N\rangle _T]^{1/2} + 1) \leq e^{K^2T/2}\left (2\sqrt{C'_{K, T}} + 1\right ).  
\qedhere 
\end{eqnarray*}
\end{proof}

\begin{proof}[Proof of Proposition \ref {prop_unique_existence}]
It suffices to show the existence and uniqueness of process $(S_r)_{r\leq t}$ for each given $(x_r)_r\in \mathcal {A}_t(x)$ and $s > 0$. \\
{\it Step 1.} For each $n\in \mathbb {N}$, define 
\begin{eqnarray*}
\tau _n = \inf \left \{ r \geq 0\ ; \ \int ^r_0g(x_v)dv \geq n \right \}\wedge t
\end{eqnarray*}
and put $x^n_r = x_r1_{[0, \tau _n]}(r)$. 
Then we can show that there is a unique solution $(Y^n_r)_r$ to the following SDE by the standard argument: 
\begin{eqnarray*}
dY^n_r = b(Y^n_r)dr + \sigma (Y^n_r)dB_r - g(x^n_r)dr, \ \ Y^n_0 = \log s. 
\end{eqnarray*}
Ito's formula implies that the process $S^n_r := \exp (Y^n_r)$ satisfies 
\begin{eqnarray*}
dS^n_r = \hat{b}(S^n_r)dr + \hat{\sigma }(S^n_r)dB_r - S^n_rg(x^n_r)dr, \ \ S^n_0 = s. 
\end{eqnarray*}
We see that $\tau _n \leq \tau _m$ and $S^n_r = S^m_r$, $r\in [0, \tau _n]$ a.s.~for each $n < m$. 
Therefore, we can define $S^\infty _r = \lim _{n\rightarrow \infty }S^n_r$ for each $r\in [0, \tau )\cap [0, t]$ a.s., 
where $\tau = \lim _{n\rightarrow \infty }\tau _n$. 

Next, we show that $\lim _{r\rightarrow \tau }S^\infty _r = 0$ a.s.~on $\{\tau \leq t\}$. 
For each $\delta > 0$, we see that 
\begin{eqnarray*}
0\leq S^\infty _{\tau - \delta } = \lim _{n\rightarrow \infty }S^n_{\tau - \delta } \leq sD_tG_\delta  \ \ \mbox { on } \{\tau \leq t\}, 
\end{eqnarray*}
where 
\begin{eqnarray*}
D_t &=& \liminf _{n\rightarrow \infty }\sup _{0\leq r\leq t}\exp \left( \int ^r_0b(Y^n_v)dv + \int ^r_0\sigma (Y^n_v)dB_v \right) , \\
G_\delta  &=& \exp \left( -\int ^{\tau - \delta }_0g(x_r)dr \right) . 
\end{eqnarray*}
Because $b$ and $\sigma $ are bounded, Lemma \ref {lem_moment} implies that $\E [D_t] < \infty $, 
hence $D_t < \infty $ a.s. 
Moreover, based on the definition of $\tau$, it holds that $G_\delta 1_{\{\tau \leq t\}}\longrightarrow 0$, $\delta \rightarrow 0$ a.s. 
Thus, we have $\lim_{\delta \rightarrow 0}S^\infty _{\tau - \delta } = 0$ a.s.~on $\{\tau \leq t\}$. 

Therefore, we can define $S_r := S^\infty _{r\wedge \tau }$ as a continuous process on $[0, t]$, and it holds that 
\begin{eqnarray*}
&&s + \int ^r_0\hat{\sigma }(S_v)dB_v + \int ^r_0(\hat{b}(S_v) - S_vg(x_v))dv\\
&=& 
s + \int ^{r\wedge \tau }_0\hat{\sigma }(S^\infty _v)dB_v + 
\int ^{r\wedge \tau }_0(\hat{b}(S^\infty _v) - S^\infty _vg(x_v))dv\\
&=& 
\lim _{n\rightarrow \infty }S^n_{r\wedge \tau _n} = S_r, \ \ r\leq t. 
\end{eqnarray*}
Thus, $(S_r)_r$ satisfies (\ref {notation_SDE}). \\
{\it Step 2.} Next, we show the uniqueness of the solution to (\ref {notation_SDE}). 
Assume that $(\tilde{S}_r)_r$ satisfies (\ref {notation_SDE}) and $\tilde{S}_0 = s$. 
We see that $Y_{r\wedge \tau _n} = \log S_{r\wedge \tau _n}$ and 
$\tilde{Y}_{r\wedge \tau _n} = \log \tilde{S}_{r\wedge \tau _n}$ satisfy 
(\ref {SDE_Y}). 
Because $b$ and $\sigma $ are Lipschitz continuous, 
we have $\E [\sup _{0\leq r\leq t}|Y_{r\wedge \tau _n} -  \tilde{Y}_{r\wedge \tau _n}|^2] = 0$. 
This implies that 
$S_{r\wedge \tau _n} = \tilde{S}_{r\wedge \tau _n}$, $r\leq t$ a.s. 
Then, we have 
$S_{r\wedge \tau _n} = \tilde{S}_{r\wedge \tau _n}$, $r\leq t$, a.s. 
Letting $n\rightarrow \infty $, we arrive at $S_{r\wedge \tau } = \tilde{S}_{r\wedge \tau }$, $r\leq t$ a.s. 
Based on (\ref {notation_SDE}), $S_r = \tilde{S}_r = 0$ for each $r$ larger than $\tau $ a.s.~on $\{\tau \leq t\}$, 
so we conclude that $(S_r)_r$ is equal to $(\tilde{S}_r)_r$ a.s. 
\end{proof}

\begin{proof}[Proof of Proposition \ref {prop_compare_value_functions}]
Because $J(t, c, x, s) \geq J^\infty (t, c, x, s)$ is clear, 
we may prove the opposite inequality. 

Fix any $(x_r)_r\in \mathcal {A}_t(x)$ and denote by $(C_r, X_r, S_r)_{r\leq t}$ its controlled process. 
Take any $K > 0$ and set $x^K_r = x_r\wedge K$. 
Then, $(x^K_r)_r\in \mathcal {A}^\infty _t(x)$ holds. 
Let $(C^K_r, X^K_r, S^K_r)$ be the controlled process of $(x^K_r)_r$. 
Then, we have $X^K_t \geq X_t$. 
Moreover, Proposition 5.2.18 in \cite {Karatzas-Shreve} implies that $S^K_r\geq S_r$, $r\leq t$ a.s. 
Therefore, it holds that 
\begin{eqnarray*}
C^K_t = c + \int ^t_0x^K_rS^K_rdr \geq c + \int ^t_0x^K_rS_rdr\ \ a.s., 
\end{eqnarray*}
and the monotone convergence theorem tells us that $\liminf _{K\rightarrow \infty }C^K_t \geq C_t$ a.s. 
Because $u\in \mathcal {C}$, we have 
$u(C_t, X_t, S_t) \leq \liminf _{K\rightarrow \infty }u(C^K_t, X^K_t, S^K_t)$. 
Then, we apply Fatou's lemma to see that 
\begin{eqnarray*}
\E [u(C_t, X_t, S_t)]\leq \liminf _{K\rightarrow \infty }\E [u(C^K_t, X^K_t, S^K_t)] \leq J^\infty (t, c, x, s). 
\end{eqnarray*}
Because $(x_r)_r\in \mathcal {A}_t(x)$ is arbitrary, we complete the proof. 
\end{proof}

To prove Theorem \ref {th_viscosity}, we define $F : D\times \mathbb {R}^3\times \mathscr {S}\longrightarrow \mathbb {R}\cup \{-\infty \}$ by 
\begin{eqnarray*}
F(z, p, \Sigma ) &=& -\frac{1}{2}\hat{\sigma }(s)^2\Sigma _{ss} - \hat{b}(s)p_s + H(s, p), \\
H(s, p) &=& \inf _{y\geq 0}f(y ; s, p), \\
f(y ; s, p) &=& sp_sg(y) - (sp_c - p_x)y, 
\end{eqnarray*}
where $\mathscr {S}\subset \mathbb {R}^3\otimes \mathbb {R}^3$ is the set of three-dimensional real symmetric matrices, and we denote 
\begin{eqnarray*}
z = 
\left(
\begin{array}{c}
 	c\\
 	x\\
 s
\end{array}
\right), \ \ 
p = 
\left(
\begin{array}{c}
 	p_c\\
 	p_x\\
 p_s
\end{array}
\right), \ \ 
\Sigma  = 
\left(
\begin{array}{ccc}
 	\Sigma _{cc} & \Sigma _{cx} & \Sigma _{cs}\\
 	\Sigma _{xc} & \Sigma _{xx} & \Sigma _{xs}\\
 	\Sigma _{sc} & \Sigma _{sx} & \Sigma _{ss}
\end{array}
\right). 
\end{eqnarray*}
Note that (\ref {HJB}) is equivalent to 
\begin{eqnarray}\label{HJB2}
\frac{\partial }{\partial t}J + F(z, \mathscr {D}J, \mathscr {D}^2J) = 0. 
\end{eqnarray}
Moreover, put 
\begin{eqnarray*}
\mathscr {U} &=& \left \{ (z, p, \Sigma )\in \tilde{D}\times \mathbb {R}^3\times \mathscr {S}\ ; \ F(z, p, \Sigma ) > -\infty  \right \}, \\
\mathscr {R} &=& \tilde{D}\times (\mathbb {R}^2\times (0, \infty ))\times \mathscr {S}. 
\end{eqnarray*}
Note that $p_s\geq 0$ holds for each $(z, p, \Sigma )\in \mathscr {U}$. 
Moreover, we have $\mathscr {R}\subset \mathscr {U}$. 

\begin{lemma}\label{lem_est_H} 
For each $(z, p, \Sigma )\in \mathscr {R}$, we have 
\begin{eqnarray}\label{est_H}
H(s, p) = 
f(h^{-1}(\mathcal {H}(s, p)\vee h(x_0)) ; s, p) \wedge 0 = f(\Xi (s, p) ; s, p),
\end{eqnarray}
where $\mathcal {H}(s, p)$ and $\Xi (s, p)$ are given by $(\ref {def_H})$--$(\ref {def_Xi})$. 
In particular, $F$ is continuous on $\mathscr {R}$. 
\end{lemma}

\begin{proof}
First, we note that $H(s, p) = sp_s\bar{H}(\mathcal {H}(s, p))$, where 
\begin{eqnarray*}
\bar{H}(\bar{y}) = \inf _{y\geq 0}\bar{f}(y ; \bar{y}), \ \ \bar{f}(y ; \bar{y}) = g(y) - \bar{y}y. 
\end{eqnarray*}
We see that 
\begin{eqnarray}\label{est_bar_H}
\bar{H}(\bar{y}) = \bar{f}(h^{-1}(\bar{y}\vee h(\bar{x}_0)) ; \bar{y})\wedge 0. 
\end{eqnarray}
Indeed, if $\bar{y}\leq h(\bar{x}_0)$, we observe 
\begin{eqnarray*}
\frac{\partial }{\partial y}\bar{f}(y ; \bar{y}) = h(y) - \bar{y} \geq h(y) - h(\bar{x}_0)\geq 0, \ \ y\geq 0 
\end{eqnarray*}
by [A3]. 
Thus we get $\bar{H}(\bar{y}) = \bar{f}(0 ; \bar{y}) = 0$. 
Moreover, [A3] also implies 
\begin{eqnarray*}
\bar{f}(h^{-1}(h(\bar{x}_0)) ; \bar{y}) = \bar{f}(\bar{x}_0 ; \bar{y}) = 
\int ^{\bar{x}_0}_0h(y')dy' - \bar{y}\bar{x}_0\geq (h(\bar{x}_0) - \bar{y})\bar{x}_0\geq 0; 
\end{eqnarray*}
hence, (\ref {est_bar_H}) holds. 
In contrast, if $\bar{y} > h(\bar{x}_0)$, we see that 
$\bar{f}(\cdot  \ ; \bar{y})$ attains the minimum at $h^{-1}(\bar{y})$ or $0$. 
When $\bar{f}(h^{-1}(\bar{y}) ; \bar{y}) < 0$, it holds that $\bar{H}(\bar{y}) = \bar{f}(h^{-1}(\bar{y}) ; \bar{y})$. 
When $\bar{f}(h^{-1}(\bar{y}) ; \bar{y}) \geq 0$, it holds that $\bar{H}(\bar{y}) = 0$. 
In both cases, we see that (\ref {est_bar_H}) actually holds. 
(\ref {est_bar_H}) implies the first equality of (\ref {est_H}). 
The second equality of (\ref {est_H}) is obtained by a straightforward calculation using [A3]. 
The last assertion is obtained by the continuity of $\hat{b}$, $\hat{\sigma }$, $h$, $\bar{H}$, and $\mathcal {H}(s, p)$. 
\end{proof}

The following proposition is obtained by a standard argument (see \cite {Fleming-Soner, Krylov, Nagai} for details). 

\begin{proposition}\label{prop_supersol}$J$ is the viscosity supersolution of $(\ref {HJB})$. 
\end{proposition}

\begin{proposition}\label{prop_subsol}Assume $(\ref {ass_viscosity})$. Then, $J$ is the viscosity subsolution of $(\ref {HJB})$. 
\end{proposition}

\begin{proof} 
Fix each $(t, z)\in (0, T]\times \tilde{D}$. 
Let $v\in C^{1, 2}((0, T]\times \tilde{D})$ be a test function, such that 
$J - v$ attains the local maximum, $0$, at $(t, z)$. 
Then, we can find an $r > 0$ such that 
\begin{eqnarray}\label{LM}
J(t', z') < v(t', z')
\end{eqnarray}
holds for each $(t', z')\in \bar{B}_r((t, z)) \setminus \{(t , z)\}$, where 
\begin{eqnarray*}
\bar{B}_r((t, z)) = \{ (t', z')\in (0, T]\times \tilde{D}\ ; \ |t'-t|^2 + |z' - z|^2 \leq r^2 \}. 
\end{eqnarray*}

For each $L > 0$, define 
\begin{eqnarray*}
F^L(z, p, \Sigma ) &=& -\frac{1}{2}\hat{\sigma }(s)^2\Sigma _{ss} - \hat{b}(s)p_s + \inf _{0\leq y\leq L}f(y ; s, p), \\
J^L(t, c, x, s) &=& \sup _{(x_r)_r\in \mathcal {A}^L_t(x)}\E [u(C_t, X_t, S_t)], \\
\mathcal {A}^L_t(x) &=& \{ (x_r)_r\in \mathcal {A}_t(x)\ ; \ |x_r|\leq L \}. 
\end{eqnarray*}
Here, $(C_r, X_r, S_r)_r$ is given as (\ref {notation_SDE}) and $(C_0, X_0, S_0) = (c, x, s)$. 
Note that $J^L(t, c, x, s)\allowbreak \nearrow J(t, c, x, s)$, $L\rightarrow \infty $. 
By the same argument as Proposition B.18 in \cite {Kato_FS}, we see that $J^L$ is a viscosity solution of 
\begin{eqnarray}\label{HJB_L}
\frac{\partial }{\partial t}J^L + F^L(z, \mathscr {D}J^L, \mathscr {D}^2J^L) = 0. 
\end{eqnarray}

Because $J^L - v$ is continuous, $J^L - v$ attains a maximum on the set $\bar{B}_r((t, z))$; namely, 
there is a $(t_L, z_L)\in \bar{B}_r((t, z))$ such that $\max _{\bar{B}_r(t, z)} = J^L(t_L, z_L) - v(t_L, z_L)$. 
We show that 
\begin{eqnarray}\label{conv_tz}
(t_L, z_L)\longrightarrow (t, z), \ \ L\rightarrow \infty . 
\end{eqnarray} 

Because $\{ (t_L, z_L) \} _L$ is a bounded sequence, we see that 
for each increasing sequence, $(L_n)_n\subset (0, \infty )$, there is a subsequence, $(L_{n_k})_k$, such that 
$(t_{L_{n_k}}, z_{L_{n_k}})$ converges to a point, $(t^*, z^*)\in \bar{B}_r(t, z)$. 
Dini's theorem implies that $J^L \longrightarrow J$, $L\rightarrow \infty $ is uniform convergence on any compact set; thus, we see that 
\begin{eqnarray*}
J(t^*, z^*) - v(t^*, z^*) = \lim _{k\rightarrow \infty }(J^{L_{n_k}}(t_{L_{n_k}}, z_{L_{n_k}}) - v((t_{L_{n_k}}, z_{L_{n_k}}))) = 0. 
\end{eqnarray*}
Combining this with (\ref {LM}), we conclude that $(t^*, z^*)$ must coincide with $(t, z)$. 
Therefore, (\ref {conv_tz}) is true. 

Next, we define $\tilde{v}\in C^{1, 2}((0, T]\times \tilde{D})$ by 
\begin{eqnarray*}
\tilde{v}(t', z') = v(t', z') + J^L(t_L, z_L) - v(t_L, z_L). 
\end{eqnarray*}
Then, we see that $J^L - \tilde{v}$ attains a local maximum, $0$, at $(t_L, z_L)$. 
Moreover, because $J^L$ is a viscosity solution to (\ref {HJB_L}), it holds that 
\begin{eqnarray}\nonumber 
&&\frac{\partial }{\partial t}\tilde{v}(t_L, z_L) + F^L(z_L, \mathscr {D}\tilde{v}(t_L, z_L), \mathscr {D}^2\tilde{v}(t_L, z_L)) \\\label{est_L2}
&=& 
\frac{\partial }{\partial t}v(t_L, z_L) + F^L(z_L, \mathscr {D}v(t_L, z_L), \mathscr {D}^2v(t_L, z_L))
\leq 0. 
\end{eqnarray}
Note that $(z_L, \mathscr {D}v(t_L, z_L), \mathscr {D}^2v(t_L, z_L))\in \mathscr {R}$ holds for large enough $L$. 
Indeed, (\ref {ass_viscosity}) implies that $(\partial /\partial s)v(t, z) > 0$, 
and the convergence $(t_L, z_L)\longrightarrow (t, z)$ and the continuity of $\mathscr {D}v$ lead us to 
$(\partial /\partial s)v(t_L, z_L) > 0$ for large enough $L$. 
Moreover, using Lemma \ref {lem_est_H} and Dini's theorem again, we see that 
$F^L$ converges to $F$ as $L\rightarrow \infty$ uniformly on any compact set in $\mathscr {R}$. 
Therefore, taking $L\rightarrow \infty $ in (\ref {est_L2}), we arrive at 
\begin{eqnarray*}
\frac{\partial }{\partial t}v(t, z) + F(z, \mathscr {D}v(t, z), \mathscr {D}^2v(t, z)) \leq 0. 
\end{eqnarray*}
This completes the proof. 
\end{proof}

\begin{proof}[Proof of Theorem \ref {th_viscosity}]
Assertion (i) is a consequence of Propositions \ref {prop_supersol}--\ref {prop_subsol}. 
Assertion (ii) is obtained by the same arguments as the proofs of Propositions B.21--B.23 in \cite {Kato_FS}. 
We note that Proposition B.22 of \cite{Kato_FS} requires the convexity of $g$ only on $[x_1, \infty )$ for large enough $x_1$. 
\end{proof}

We prepare the following lemma to show Theorem \ref {th_S-shaped}. 

\begin{lemma}\label{lem_S-shaped} \ Assume {\rm [B]} and that $J\in C^{1,1,1,2}((0, T]\times D)$. 
It holds that 
\begin{eqnarray}
\label{diff_c}
&&\frac{\partial }{\partial c}J(t, c, 0, s) = U'(c), \\\label{diff_x}
&&\frac{\partial }{\partial x}J(t, c, 0, s) \geq  sU'(c), \\\label{diff_s}
&&\frac{\partial }{\partial s}J(t, c, 0, s) = \frac{\partial }{\partial x}J(t, c, x, 0) = 0, \\\label{diff_t}
&&\frac{\partial }{\partial t}J(t, c, 0, s) = \frac{\partial }{\partial t}J(t, c, x, 0) = 0 
\end{eqnarray}
for each $t > 0$ and $(c, x, s)\in D$. 
\end{lemma}

\begin{proof}
(\ref {diff_c}), (\ref {diff_s}) and (\ref {diff_t}) are obtained from $J(t, c, 0, s) = J(t, c, x, 0) = U(c)$. 

To show (\ref {diff_x}), for each fixed $t > 0$ and each $x\in (0, t^2)$, 
set $x_r = \sqrt{x}1_{[0, \sqrt{x}]}(r)$ and let $(C_r, X_r, S_r)_r$ be the controlled process associated with $(x_r)_r\in \mathcal {A}_t(x)$. 
Then, we see that 
\begin{eqnarray}\nonumber 
&&\frac{1}{x}(J(t, c, x, s) - J(t, c, 0, s))\\\label{Cl1}
&\geq &
\frac{1}{x}\E [U(C_t) - U(c)] = \E \left[ \int ^1_0U'(c + kxA_x)dkA_x\right] , 
\end{eqnarray}
where $A_x = \frac{1}{\sqrt{x}}\int ^{\sqrt{x}}_0S_rdr$. 
By the Doob inequality, (3.18) in \cite {Ishitani-Kato_COSA1}, and (\ref {est_S_Z})--(\ref {est_Z_hat}), we have 
$\E [|A_x - s|] \longrightarrow 0$, \ $x\rightarrow 0$. 
In particular, $A_x$ converges to $s$ in probability. 
Moreover, by (\ref {est_S_Z})--(\ref {est_Z_hat}) and the concavity of $U$, we have 
\begin{eqnarray*}
\E [\sup _{0\leq x\leq t^2}A^2_x] < \infty , \ \ 
\E [\sup _{0\leq k\leq 1, 0\leq x\leq t^2}(U'(c + kxA_x))^2] \leq (U'(c))^2. 
\end{eqnarray*}
Therefore, we can apply the dominated convergence theorem to get 
\begin{eqnarray}\label{Cl2}
\E \left[ \int ^1_0U'(c + kxA_x)dkA_x\right] \longrightarrow sU'(c), \ \ x\rightarrow 0. 
\end{eqnarray}
(\ref {Cl1})--(\ref {Cl2}) lead us to (\ref {diff_x}). 
\end{proof}

\begin{proof}[Proof of Theorem \ref {th_S-shaped}]
Define 
\begin{eqnarray}\label{def_x_hat}
\hat{x}_t = \Xi (\overline {S}_t, \mathscr {D}J(T - t, \overline {C}_t, \overline {X}_t, \overline {S}_t))1_{[0, \bar{\tau })}(t). 
\end{eqnarray}
Because $\overline {X}_{T\wedge \bar{\tau }}\geq 0$, it holds that $(\hat{x}_t)_t\in \mathcal {A}_T(x_0)$. 
Then, Proposition \ref {prop_unique_existence} implies that there is a controlled process $(\hat{C}_t, \hat{X}_t, \hat{S}_t)_t$ associated with $(\hat{x}_t)_t$. 
We see that $\hat{C}_t = \overline {C}_{t\wedge \bar{\tau }}$, $\hat{X}_t = \overline {X}_{t\wedge \bar{\tau }}$ and $\hat{S}_{t\wedge \bar{\tau }} = \overline {S}_{t\wedge \bar{\tau }}$. 

Put $\hat{\tau }_R = \inf \{ t\geq 0 \ ; \ \hat {S}_{t} \geq R \}\wedge (T-1/R)_+$ for each $R > 0$. 
Note that (\ref {est_S_Z})--(\ref {est_Z_hat}) and the Chebyshev inequality imply that 
$\hat{\tau }_R\nearrow T$, $R\rightarrow \infty $. 
Ito's formula gives us
\begin{eqnarray}\nonumber 
&&\E [J(T - \hat{\tau }_R, \hat{C}_{\hat{\tau }_R}, \hat{X}_{\hat{\tau }_R}, \hat{S}_{\hat{\tau }_R})] - J(T, c_0, x_0, s_0) \\\label{est_S-shaped1}
&=& 
\E \left [ \int ^{\hat{\tau }_R}_0\left(  
-\frac{\partial }{\partial t}J + \mathscr {L}^{\hat{x}_t}J\right) (T - t, \hat{C}_t, \hat{X}_t, \hat{S}_t)dt\right] . 
\end{eqnarray}
Based on Theorem \ref {th_viscosity}(i), Lemma \ref {lem_est_H}, and the smoothness of $J$, we have 
\begin{eqnarray}\nonumber 
&&\left( 
-\frac{\partial }{\partial t}J + \mathscr {L}^{\hat{x}_t}J\right) (T - t, \hat{C}_t, \hat{X}_t, \hat{S}_t)\\\label{est_S-shaped2}
&=& 
\left( 
-\frac{\partial }{\partial t}J + \sup _{y\geq 0}\mathscr {L}^yJ\right) (T - t, \overline {C}_t, \overline {X}_t, \overline {S}_t) = 0, \ \ t < \bar{\tau }. 
\end{eqnarray}
On $\{ t\geq \bar{\tau } \}$, either $\hat{X}_t = 0$ or $\hat{S}_t = 0$ holds. 
If $\hat{X}_t = 0$, we have 
\begin{eqnarray}\nonumber 
&&\left( 
-\frac{\partial }{\partial t}J + \sup _{y\geq 0}\mathscr {L}^yJ\right) (T - t, \hat{C}_t, 0, \hat{S}_t)\\\nonumber 
&=& 
\sup _{y\geq 0}\left\{ \left( \hat{S}_tU'(\hat{C}_t) - \frac{\partial }{\partial x}J(T - t, \hat{C}_t, 0, \hat{S}_t) \right) y\right\} \\\label{est_S-shaped3}
&=& 
0 = 
\left( 
-\frac{\partial }{\partial t}J + \mathscr {L}^{\hat{x}_t}J\right) (T - t, \hat{C}_t, 0, \hat{S}_t) , \ \ t\geq \bar{\tau }
\end{eqnarray}
by Lemma \ref {lem_S-shaped}. 
If $\hat{S}_t = 0$, Lemma \ref {lem_S-shaped} also implies that 
\begin{eqnarray}\nonumber 
&&\left( 
-\frac{\partial }{\partial t}J + \sup _{y\geq 0}\mathscr {L}^yJ\right) (T - t, \hat{C}_t, \hat{X}_t, 0)\\\label{est_S-shaped4}
&=& 
0 = 
\left( 
-\frac{\partial }{\partial t}J + \mathscr {L}^{\hat{x}_t}J\right) (T - t, \hat{C}_t, \hat{X}_t, 0), \ \ t\geq \bar{\tau }.  
\end{eqnarray}
Combining (\ref {est_S-shaped1})--(\ref {est_S-shaped4}), we arrive at 
$\E [J(T - \hat{\tau }_R, \hat{C}_{\hat{\tau }_R}, \hat{X}_{\hat{\tau }_R}, \hat{S}_{\hat{\tau }_R})] = J(T, c_0, x_0, s_0)$. 
Letting $R\rightarrow \infty $, we have 
$\E [u(\hat{C}_T, \hat{X}_T, \hat{S}_T)] = J(T, c_0, x_0, s_0)$ 
due to the dominated convergence theorem. 
Therefore, $(\hat{x}_t)_t$ is the optimizer to $J(T, c_0, x_0, s_0)$. 
Based on (\ref {def_Xi}) and (\ref {def_x_hat}), we see that $\hat{x}_t = 0$ or $\hat{x}_t > \bar{x}_0$. 
\end{proof}

\begin{proof}[Proof of Theorem \ref {th_verification}]
Fix any $(x_t)_t\in \mathcal {A}^\infty _T(x_0)$ and $n\in \mathbb {N}$, and let $x^n_t = (1-1/n)x_t$. 
Denote by $(C_t, X_t, S_t)_t$ (resp., $(C^n_t, X^n_t, S^n_t)$) 
the controlled process associated with $(x_t)_t$ (resp., $(x^n_t)_t$). 
Note that (\ref {est_S_Z})--(\ref {est_Z_hat}) imply 
$\E [\sup _{0\leq t\leq T}(S^n_t)^m] < \infty $ 
for each $m > 0$, and that $(C^n_t, X^n_t, S^n_t) \in \tilde{D}$, $t\in [0, T]$ holds a.s. 

Take any $R > 0$ and set $\tau _R = \inf \{ t\geq 0 \ ; \ S^n_t\geq R\} \wedge (T-1/R)_+$. 
By a standard argument using Ito's formula, we arrive at 
\begin{eqnarray*}
&&\E [v(T - \tau _R, C^n_{\tau _R}, X^n_{\tau _R}, S^n_{\tau _R})] - v(T, c_0, x_0, s_0)\\
&\leq & 
\E \left[ \int ^{\tau _R}_0\left( -\frac{\partial }{\partial t}v + \sup _{y\geq 0}\mathscr {L}^yv \right) 
(T - t, C^n_t, X^n_t, S^n_t)dt \right] . 
\end{eqnarray*}
Combining this with assumption (iii), we have 
\begin{eqnarray*}
\E [v(T - \tau _R, C^n_{\tau _R}, X^n_{\tau _R}, S^n_{\tau _R})] \leq  v(T, c_0, x_0, s_0). 
\end{eqnarray*}
Here, based on assumptions (i)--(ii) and the Chebyshev inequality, we see that 
$\tau _R\nearrow T$, $R\rightarrow \infty $ a.s.~and 
\begin{eqnarray}\label{est_verification3}
\E [u(C^n_T, X^n_T, S^n_T)] \leq v(T, c_0, x_0, s_0). 
\end{eqnarray}
Because $(x_t)_t$ is essentially bounded, by using Theorem 2.5.9 in \cite {Krylov}, we obtain 
$\E [\sup _{0\leq t\leq T}|\log S^n_t - \log S_t|^4] \longrightarrow 0$, \ $n\rightarrow \infty $. 
Then, we have $\E [\sup _{0\leq t\leq T}|S^n_t - S_t|^2] \longrightarrow 0$, and thus 
$\E [|C^n_T - C_T|]\longrightarrow 0$. 
Moreover, we see that $\E [|X^n_T - X_T|]\longrightarrow 0$. 
Therefore, letting $n\rightarrow \infty $ in (\ref {est_verification3}) 
and applying Lemma B.2 in \cite {Kato_FS}, we arrive at 
\begin{eqnarray*}
\E [u(C_T, X_T, S_T)] \leq v(T, c_0, x_0, s_0). 
\end{eqnarray*}
Because $(x_t)_t\in \mathcal {A}^\infty _T(x_0)$ is arbitrary, we deduce that 
\begin{eqnarray*}
J(T, c_0, x_0, s_0) = J^\infty (T, c_0, x_0, s_0) \leq v(T, c_0, x_0, s_0). 
\end{eqnarray*}
This and assumption (iv) lead to the conclusion that 
\begin{eqnarray*}
J(T, c_0, x_0, s_0) = \E [u(\hat{C}_T, \hat{X}_T, \hat{S}_T)] = v(T, c_0, x_0, s_0).  \qedhere 
\end{eqnarray*}
\end{proof}

Theorems \ref {th_mixed_power} and \ref {th_TWAP} are obtained by a straightforward calculation using Theorem \ref {th_verification}. 

\begin{proof}[Proof of Theorem \ref {th_TWAP_Levy}]
Theorem 5.2 in \cite {Ishitani-Kato_COSA2} tells us that (\ref {def_J_Levy}) is equivalent with the optimization problem
\begin{eqnarray*}
c_0 + \sup _{(x_t)_t\in \mathcal {A}_T(x_0)}\int ^T_0\hat{S}_tx_tdt, 
\end{eqnarray*}
where 
\begin{eqnarray*}
&&d\hat{S}_t = -\hat{S}_t(\tilde{\mu } + \hat{g}(x_t))dt, \ \ \hat{S}_0 = s_0, \\
&&\hat{g}(x) = \gamma \alpha _0x^2 + \alpha _1\log (\alpha _0\beta _1x^2 + 1). 
\end{eqnarray*}
Moreover, we see that 
\begin{eqnarray*}
\hat{g}''(x) \geq  \frac{\alpha _0\alpha _1\beta _1(\alpha _0\beta _1x^2 - 3)^2}{4(\alpha _0\beta _1x^2 + 1)^2}, 
\end{eqnarray*}
hence, $\hat{g}(x)$ is strictly convex (see Corollary 1.3.10 in \cite {Niculescu-Persson} for instance). 
Therefore, we can apply Theorem \ref {th_TWAP} to complete the proof. 
The optimal execution speed $\hat{\nu }$ satisfies $G_{\hat{g}'}(\hat{\nu }) = \tilde{\mu }$, where 
\begin{eqnarray*}
G_{\hat{g}'}(x) = x\hat{g}'(x) - \hat{g}(x) = 
\gamma \alpha _0x^2 + 
\alpha _1\left\{ 2\left( 1 - \frac{1}{1 + \alpha _0\beta _1x^2}\right) - \log (\alpha _0\beta _1x + 1) \right\} .   \qedhere 
\end{eqnarray*}
\end{proof}


\begin{thebibliography}{0}

\bibitem 
{Alfonsi-Fruth-Schied}Alfonsi, A., Fruth, A., and Schied, A.: 
Optimal execution strategies in limit order books with general shape functions, 
{\em Quant. Finance} {\bf 10} (2010) 143--157. 


\bibitem 
{Almgren-Chriss}Almgren, R. F. and Chriss, N.: 
Optimal execution of portfolio transactions,
{\em J. Risk}, \textbf{18} (2000) 57--62.

\bibitem 
{Ane-Geman}An\'{e}, T. and Geman, H.: 
Order flow, transaction clock, and normality of asset returns, 
{\em J. Finance}, \textbf{55} (2000), no. 5, 2259--2284. 

\bibitem 
{Bertsimas-Lo}Bertsimas, D. and Lo, A. W.: Optimal control of execution costs,
{\em J. Financ. Mark.}, \textbf{1} (1998) 1--50.


\bibitem 
{DaLio-Ley}Da Lio, F., Ley, O.: 
Convex Hamilton--Jacobi equations under superlinear growth conditions on data. 
{\em Appl. Math. Optim.} \textbf {63} (2011), no. 3, 309--339. 

\bibitem 
{Fleming-Soner}Fleming, W. H. and Soner, H. M.: 
{\em Controlled Markov Processes and Viscosity Solutions,} 
Springer, New York, 1992. 


\bibitem 
{Gatheral-Schied_AC1}Gatheral, J. and Schied, A.: 
Optimal trade execution under geometric Brownian motion in the Almgren and Chriss framework, 
{\em Int. J. Theor. Appl. Finance} \textbf{14} (2011), no. 3, 353--368. 

\bibitem 
{Gatheral-Schied}Gatheral, J. and Schied, A.: 
Dynamical models of market impact and algorithms for order execution,
{\em Handbook on Systemic Risk, Eds. Fouque, J. P. and Langsam, J.}, Cambridge University Press, Cambridge, 2013, 579--602.

\bibitem 
{Geman}Geman, H.: 
Stochastic clock and financial markets, 
{\em Aspects of Mathematical Finance, Eds. Yor, M.}, Springer, New York, 2008, 37--52. 

\bibitem 
{Gueant}Gu\'eant, O.: 
Permanent market impact can be nonlinear, 
{\em Preprint} (2014).  

\bibitem 
{Ishitani-Kato_COSA1}Ishitani, K. and Kato, T.: 
Mathematical formulation of an optimal execution problem with uncertain market impact, 
{\em Commun. Stoch. Anal.} \textbf{9} (2015), no. 3, 113--129. 

\bibitem 
{Ishitani-Kato_COSA2}Ishitani, K. and Kato, T.: 
Theoretical and numerical analysis of an optimal execution problem with uncertain market impact, 
{\em Commun. Stoch. Anal.} \textbf{9} (2015), no. 3, 343--366. 

\bibitem 
{Karatzas-Shreve}Karatzas, I. and Shreve, S. E.: 
{\em Brownian Motion and Stochastic Calculus 2nd. edition}, 
Springer, New York, 1991. 

\bibitem 
{Kato_FS}Kato, T.: 
An optimal execution problem with market impact,
{\em Finance Stoch.} \textbf {18} (2014), no. 3, 695--732.

\bibitem 
{Kato_TJSIAM}Kato, T.: 
Non-linearity of market impact functions:
empirical and simulation-based studies on convex/concave market impact functions and
derivation of an optimal execution model,
{\em Trans. Jpn. Soc. Ind. Appl. Math.} 
\textbf{24} (2014), no. 3, 203--237 (in Japanese).

\bibitem 
{Kato_JSIAM_VWAP}Kato, T.: 
VWAP execution as an optimal strategy, 
{\em JSIAM Lett.} \textbf {7} (2015) 33--36. 


\bibitem 
{Kato_AC_Preprint}Kato, T.: 
An optimal execution problem in the volume-dependent Almgren--Chriss model, 
{\em Preprint} (2017). 

\bibitem 
{Krylov}Krylov, N. V.: 
{\em Controlled Diffusion Processes,} 
Springer, Berlin, 1980. 

\bibitem 
{Madhavan}Madhavan, A.: 
VWAP Strategies, 
{\em Trading} \textbf {1} (2002) 32--39. 

\bibitem 
{Mazur}Mazur, S.: 
Modeling market impact and timing risk in volume time, 
{\em Algorithmic Finance} \textbf {2} (2013), no. 2, 113--126. 

\bibitem 
{Nagai}Nagai, H.: 
{\em Stochastic differential equations,} 
Kyoritsu Shuppan, Tokyo, 1999. 

\bibitem 
{Niculescu-Persson}Niculescu, C. and Persson, L. -E.: 
{\em Convex functions and their applications,} 
Springer, New York, 2006. 

\bibitem 
{Obizhaeva-Wang}Obizhaeva, A. and Wang, J.: 
Optimal trading strategy and supply/demand dynamics, 
{\em J. Financ. Mark.} \textbf {16} (2013), no. 1, 1--32. 

\bibitem 
{Rosu}Ro\c{s}u, I.: 
A dynamic model of the limit order book, 
{\em Rev. Financ. Stud.} \textbf {22} (2009), no. 11, 4601--4641. 

\bibitem 
{Veraat-Winkel}Veraat, E. D. and Winkel, M.: 
Time change, 
{\em Encyclopedia of Quantitative Finance, Eds. Cont, R.}, Wiley, Chichester, 2010, 1812--1816. 

\end{thebibliography}
\end{document}